\documentclass[a4paper,twocolumn,superscriptaddress,preprintnumbers,nobalancelastpage,longbibliography,accepted=2020-03-02]{quantumarticle}

\pdfoutput=1
\usepackage[numbers,sort&compress]{natbib}

\usepackage{graphicx, color, graphpap}
\usepackage{enumitem}
\usepackage{amssymb}
\usepackage{amsthm}
\usepackage{float}
\usepackage{multirow}
\usepackage[colorlinks=true,citecolor=blue,linkcolor=magenta]{hyperref}
\usepackage[T1]{fontenc}
\usepackage{bbm}
\usepackage{thmtools,thm-restate}
\usepackage{verbatim}
\usepackage{mathtools}
\usepackage{epstopdf}
\usepackage{xspace}

\usepackage[utf8]{inputenc}
\usepackage[english]{babel}
\usepackage{amsmath}

\usepackage{tikz}

\long\def\ca#1\certified bounds{} 

\newcommand{\ketbra}[2]{| \hspace{1pt} #1 \rangle \langle #2 \hspace{1pt} |}

\newcommand{\ket}[1]{|#1\rangle}               
\newcommand{\bra}[1]{\langle #1|}              
\newcommand{\dya}[1]{\ket{#1}\!\bra{#1}}
\newcommand{\matl}[3]{\langle #1|#2|#3\rangle} 

\newcommand{\DQC}{\ensuremath{\mathsf{DQC1}}\xspace}

\newcommand{\poly}{\operatorname{poly}}

\newcommand{\rank}{\text{rank}}

\newcommand{\DN}{D}
\newcommand{\EC}{\mathcal{E}}

\newcommand{\FN}{F}

\newcommand{\UC}{\mathcal{U}}

\newcommand{\ZC}{\mathcal{Z}}

\newcommand{\HS}{\text{HS}}

\newcommand{\Tr}{{\rm Tr}}

\renewcommand{\geq}{\geqslant}
\renewcommand{\leq}{\leqslant}

\renewcommand{\vec}[1]{\boldsymbol{#1}}  

\newcommand{\ad}{^\dagger}

\newcommand*{\id}{\openone}

\newcommand{\rhot}{\tilde{\rho}}

{}
{}

\newtheorem{lemma}{Lemma}
\newtheorem{proposition}{Proposition}
\newtheorem{propositionSM}{Proposition}

\theoremstyle{definition}

\begin{document}
\title{Variational Quantum Fidelity Estimation}

\author{M. Cerezo}
\affiliation{Theoretical Division, Los Alamos National Laboratory, Los Alamos, NM, USA}
\affiliation{Center for Nonlinear Studies, Los Alamos National Laboratory, Los Alamos, NM, USA
}
\orcid{0000-0002-2757-3170}
\thanks{\texttt{marcovsebastian@gmail.com}}
\author{Alexander Poremba}
\affiliation{Computing and Mathematical Sciences, California Institute of Technology, Pasadena, CA, USA.}
\orcid{0000-0002-7330-1539}
\author{Lukasz Cincio}
\affiliation{Theoretical Division, Los Alamos National Laboratory, Los Alamos, NM, USA}
\orcid{0000-0002-6758-4376}
\author{Patrick J. Coles}
\affiliation{Theoretical Division, Los Alamos National Laboratory, Los Alamos, NM, USA}

\begin{abstract}
Computing quantum state fidelity will be important to verify and characterize states prepared on a quantum computer. In this work, we propose novel lower and upper bounds for the fidelity $F(\rho,\sigma)$ based on the ``truncated fidelity'' $F(\rho_m, \sigma)$, which is evaluated for a state $\rho_m$ obtained by projecting $\rho$ onto its $m$-largest eigenvalues. Our bounds can be refined, i.e., they tighten monotonically with $m$. To compute our bounds, we introduce a hybrid quantum-classical algorithm, called Variational Quantum Fidelity Estimation, that involves three steps: (1) variationally diagonalize $\rho$, (2) compute matrix elements of $\sigma$ in the eigenbasis of $\rho$, and (3) combine these matrix elements to compute our bounds. Our algorithm is aimed at the case where $\sigma$ is arbitrary and $\rho$ is low rank, which we call low-rank fidelity estimation, and we prove that no classical algorithm can efficiently solve this problem under reasonable assumptions. Finally, we demonstrate that our bounds can detect quantum phase transitions and are often tighter than  previously known computable bounds for realistic situations.
\end{abstract}

\maketitle

\section{Introduction}

In the near future, quantum computers will become quantum state preparation factories. They will prepare ground and excited states of molecules~\cite{peruzzo2014VQE}, states that simulate quantum dynamics~\cite{brown2010using}, and states that encode the solutions to linear systems~\cite{harrow2009quantum}. These states will necessarily be impure, either intentionally (e.g., when studying thermal states) or due to incoherent noise of the quantum computer (e.g., $T_1$ and $T_2$ processes). Verification and characterization of these mixed states will be important, and hence efficient algorithms will be needed for this purpose. A widely used measure for verification and characterization is the fidelity~\cite{uhlmann1976transition, jozsa1994fidelity, nielsen2010, wilde2017,Liang_2019}. For example, one may be interested in the fidelity with a fixed target state (i.e., for verification) or the fidelity between subsystems of many-body states to study behavior near a phase transition (i.e., for characterization)~\cite{hauru2018uhlmann}. For two  states $\rho$ and $\sigma$, the fidelity is defined as \cite{uhlmann1976transition, jozsa1994fidelity, nielsen2010, wilde2017}
\begin{equation}
\FN(\rho,\sigma) = 	\Tr\sqrt{\sqrt{\rho}\sigma\sqrt{\rho}}=\| \sqrt{\rho}\sqrt{\sigma} \|_1\,,
\label{eqnFidelity}
\end{equation}
where $\|A \|_1=\Tr \sqrt{A\ad A}$. See~\cite{nielsen2010, wilde2017,Liang_2019} for properties of $\FN(\rho,\sigma)$.

Classically computing $\FN(\rho,\sigma)$, or any other metric on quantum states, could scale exponentially due to the exponentially large dimension of the density matrix. This raises the question of whether a quantum computer could avoid this exponential scaling. However, $\FN(\rho,\sigma)$ involves non-integer powers of $\rho$ and $\sigma$, implying that there is no exact quantum algorithm for computing it directly from the probability of a measurement outcome on a finite number of copies of $\rho$ and $\sigma$. In addition,  deciding whether the trace distance (which is closely related to fidelity~\cite{nielsen2010, wilde2017}) is large or small is QSZK-complete~\cite{watrous2002quantum}. Here, QSZK (quantum statistical zero-knowledge) is a complexity class that contains BQP~\cite{nielsen2010} (bounded-error quantum polynomial time). It is therefore reasonable to suspect that estimating $\FN(\rho,\sigma)$ is hard even for quantum computers.

This does not preclude the efficient estimation of fidelity for the practical case when one of the states is low rank.  Low-rank states appear in condensed matter physics~\cite{PhysRevLett.101.010504} (marginals of grounds states) and data science~\cite{bishop2006pattern} (covariance matrices). We define Low-Rank Fidelity Estimation as the problem of estimating $F(\rho, \sigma)$ when $\sigma$ is arbitrary and $\rho$ is approximately low rank. We prove under standard assumptions that a classical algorithm cannot efficiently perform Low-Rank Fidelity Estimation.

In this work, we propose a variational hybrid quantum-classical algorithm \cite{peruzzo2014VQE,farhi2014quantum,mcclean2016theory,Romero17,li2017efficient, QAQC,larose2018,arrasmith2019variational,sharma2019noise,cirstoiu2019variational} for Low-Rank Fidelity Estimation called Variational Quantum Fidelity Estimation (VQFE). VQFE computes upper and lower bounds on $F(\rho, \sigma)$ that can be refined to arbitrary tightness. Our bounds are based on the \textit{truncated fidelity}, which involves evaluating~\eqref{eqnFidelity} for $\sigma$ and $\rho_m$, a truncated version of $\rho$ obtained by projecting $\rho$ onto the subspace associated with its $m$-largest eigenvalues. Crucially, our bounds tighten monotonically in $m$, and eventually they equal the fidelity when $m=\rank(\rho)$.

This is in contrast to the state-of-the-art quantum algorithm to bound the fidelity, which employs fixed bounds called the sub- and super-fidelity bounds (SSFB) \cite{MiszczakPHUZ09,chen2002alternative,mendoncca2008alternative,Punchala2009bound}, defined as
\begin{align}
E(\rho,\sigma)&= \Tr\rho\sigma + \sqrt{2\left[(\Tr\rho\sigma)^2-\Tr\rho\sigma\rho\sigma\right]}\,,\\
G(\rho,\sigma)&= \Tr\rho\sigma+\sqrt{(1-\Tr\rho^2)(1-\Tr\sigma^2)}\,,
\end{align}
respectively, such that  $\sqrt{E(\rho,\sigma)}\leq \FN(\rho,\sigma) \leq \sqrt{G(\rho,\sigma)}$. Since the SSFB are expressed as traces of products of density matrices, they can be efficiently measured on a quantum computer \cite{ brun2004measuring,Bartkiewicz2013measuring, cincio2018learning,garcia2013swap}. These bounds are generally looser when both $\rho$ and $\sigma$ have high rank, and hence the SSFB likewise perform better when one of the states is low rank. Below we give a detailed comparison, showing that VQFE often out-performs the SSFB. We also note that VQFE only requires $2n+1$ qubits while the SSFB require $4n+1$ qubits, for $n$-qubit states.

To produce certified bounds from the output of VQFE, we prove several novel bounds that should be of independent interest in quantum information theory. In addition, we introduce the \textit{fidelity spectrum}, which is the collection of all truncated fidelities. In the same sense that the entanglement spectrum provides more information than a single entanglement measure \cite{PhysRevLett.101.010504}, we argue that the fidelity spectrum gives more information about the closeness of $\rho$ and $\sigma$ than just $F(\rho,\sigma)$.

In what follows, we first present our bounds and the VQFE algorithm. We compare its performance with the SSFB and illustrate its application to phase transitions. All proofs of our results are delegated to the Appendix.

\section{Truncated fidelity bounds}

Let $\Pi_{m}^\rho$ be the projector onto the subspace spanned by the eigenvectors of $\rho$ with the $m$-largest eigenvalues. Consider the sub-normalized states
\begin{equation}
\rho_m^{\rho}= \Pi_m^\rho\rho \Pi_m^\rho =\sum_{i=1}^m r_i \ketbra{r_i}{r_i}\,, \quad \sigma_m^\rho= \Pi_m^\rho\sigma \Pi_m^\rho\,,
\label{EqnRhoTruncated}
\end{equation}
where the eigenvalues $\{r_i\}$ of $\rho$ are in decreasing order ($r_i\geq r_{i+1}$). For simplicity we denote $\rho_m^{\rho}=\rho_m$. These sub-normalized states can be used to define the {\it truncated fidelity}
\begin{equation}
    \FN(\rho_m,\sigma_m^\rho)=\FN(\rho_m,\sigma)= \left\| \sqrt{\rho_m}\sqrt{\sigma} \right\|_1\,,
    \label{EqnTrunF}
\end{equation}
and the {\it truncated generalized fidelity} 
\begin{align}
\FN_\ast(\rho_m,\sigma_m^\rho) =&\left\| \sqrt{\rho_m}\sqrt{\sigma} \right\|_1 \notag \\
&+\sqrt{(1-\Tr\rho_m)(1-\Tr \sigma_m^\rho)}\,,
     \label{EqnTrunGF}
\end{align}
where the generalized fidelity~\cite{tomamichel2010duality,tomamichel2015quantum} was defined for two sub-normalized states as 
$\FN_\ast(\rho,\sigma)=\|\sqrt{\rho}\sqrt{\sigma}\|_1+\sqrt{(1-\Tr\,\rho)(1-\Tr\,\sigma)}$. The generalized fidelity  reduces to \eqref{eqnFidelity} if at least one state is normalized. 
Equations~\eqref{EqnTrunF} and \eqref{EqnTrunGF} are in fact lower and upper bounds, respectively, for the fidelity.
\begin{proposition}
\label{prop1}
The following truncated fidelity bounds (TFB) hold:
\begin{equation}
    \FN(\rho_m,\sigma_m^\rho)\leq \FN(\rho,\sigma)\leq\FN_\ast(\rho_m,\sigma_m^\rho)\,.
\label{EqnBounds}
\end{equation}
\end{proposition}
We refer to the collection of TFB for different values of $m$ as the {\it fidelity spectrum}. The TFB are satisfied with equality when $m =\rank(\rho)$. Moreover, they monotonically get tighter as $m$ increases. Hence they can be refined to arbitrary tightness by increasing $m$.

\begin{proposition}
\label{prop2}
The truncated fidelity $\FN(\rho_m,\sigma_m^\rho)$ is monotonically increasing in $m$, and the truncated generalized fidelity $\FN_{*}(\rho_m,\sigma_m^\rho)$ is monotonically decreasing in $m$.
\end{proposition}

Ultimately, we will consider the case when $\rho$ is either low rank or {\it $\epsilon$-low-rank}. Here we define the $\epsilon$-rank as a generalization of the rank to within some $\epsilon$ error:
\begin{equation}
\rank_{\epsilon}(\rho) = \min \{  m \in \{1,..., d\} : \|\rho - \rho_m\|_1 \leq \epsilon  \}\,,  
\label{epsilonrank}
\end{equation}
where $d$ is the Hilbert space dimension. Note that $\rank_{0}(\rho) = \rank(\rho)$. We remark that the looseness of the TFB is bounded
by the square root of $\|\rho - \rho_m\|_1 = 1 - \Tr(\rho_m)$:
\begin{align}\label{previousprop3}
\FN_\ast(\rho_m,\sigma_m^\rho) - \FN(\rho_m,\sigma_m^\rho)\leq \|\rho - \rho_m\|_1^{1/2}\,, 
\end{align}
where we have used \eqref{EqnTrunF}, \eqref{EqnTrunGF}, and $\Tr \sigma_m^\rho\leq 1$.
Hence, the TFB looseness is bounded by $\sqrt{\epsilon}$ provided that $m = \rank_{\epsilon}(\rho)$.

\begin{figure}
    \centering
    \includegraphics[width=\columnwidth]{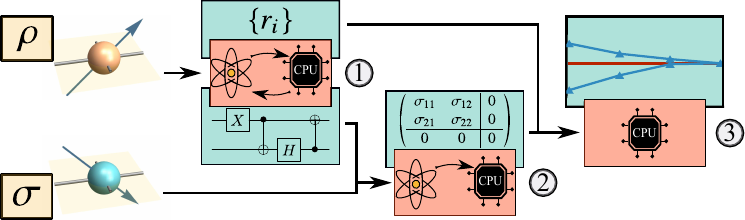}
    \caption{The VQFE algorithm. (1) First, $\rho$ is diagonalized with a hybrid quantum-classical optimization loop, outputting the largest eigenvalues $\{r_i\}$ of $\rho$ and a gate sequence that prepares the associated eigenvectors. (2) Second, a hybrid quantum-classical computation gives the matrix elements of $\sigma$ in the eigenbasis of $\rho$. (3) Finally, classical processing gives upper and lower bounds on $F(\rho, \sigma)$.}
    \label{fig:VQFE}
\end{figure}

\section{The VQFE algorithm}

Figure~\ref{fig:VQFE} shows the overall structure of the VQFE algorithm. VQFE involves three steps: (1) approximately diagonalize $\rho$ with a variational hybrid quantum-classical algorithm, (2) compute matrix elements of $\sigma$ in the eigenbasis of $\rho$, (3) classically process these matrix elements to produce certified bounds on $F(\rho, \sigma)$.

The first subroutine employs Variational Quantum State Diagonalization (VQSD)~\cite{larose2018}, a variational hybrid algorithm that
takes in two copies of $\rho$ and outputs approximations of the $m$-largest eigenvalues $\{r_i\}$ and a gate sequence $U$ that prepares the associated eigenvectors $\{\ket{r_i}\}$. This subroutine involves a  quantum-classical optimization loop that minimizes a cost function $C$ that quantifies how far $\rhot= U\rho U\ad$ is from being diagonal. Namely, $C=D_{\HS}(\rhot, \ZC(\rhot))$, with $D_{\HS}(\rho,\sigma)=\Tr[(\rho-\sigma)^2]$ the Hilbert-Schmidt distance and $\ZC$ the dephasing channel in the computational basis. When $C=0$ we have exact diagonalization, whereas for $C\neq 0$, VQSD outputs the eigenvalues and eigenvectors of 
\begin{equation}
\rho'=U\ad \ZC(\rhot) U=\sum_{i} r_i' \ketbra{r_i'}{r_i'}\,,
    \label{EqnRhoPrime}
\end{equation}
such that $C=D_{\HS}(\rho, \rho')$.

The second subroutine measures the matrix elements $\sigma_{ij}=\matl{r_i}{\sigma}{r_j}$, or more precisely $\sigma_{ij}'=\matl{r_i'}{\sigma}{r_j'}$ if $C\neq 0$. This is done by preparing superpositions of the eigenvectors of the form $(\ket{r_i'} + \ket{r_j'})/\sqrt{2}$ (by means of VQSD's eigenstate preparation circuit \cite{larose2018}) and computing the inner product of this superposition with $\sigma$ via a Swap Test. (For example, see \cite{cincio2018learning,arrasmith2019variational} for the precise circuit.)  For a fixed $m$, one needs to measure $m(m+1)/2$ matrix elements. However, when incrementing $m$ to $m+1$, one only needs to measure $2m+1$ new quantities.

The third subroutine involves only classical computation, combining the outputs of the previous subroutines to produce bounds on $F(\rho, \sigma)$. These bounds employ the TFB described before, for which one needs to compute  
\begin{equation}
\left\| \sqrt{\rho_m}\sqrt{\sigma} \right\|_1=\Tr\sqrt{\sum_{i,j} T_{ij} \ketbra{r_i}{r_j}}\,.
\label{EqnTMatrix}
\end{equation}
Here we have defined the $m\times m$  matrix $T$ with elements  $T_{ij}=\sqrt{r_ir_j}\matl{r_i}{\sigma}{r_j}=\sqrt{r_ir_j}\sigma_{ij}$ such that $T\geq 0$.  Note that the $T_{ij}$ can be computed directly from the outputs of the first two subroutines. One can then classically diagonalize $T$ to compute \eqref{EqnTMatrix}.  (If shot noise leads to a non-positive matrix $T$, one can efficiently find the maximum-likelihood matrix $\tilde{T}\geq 0$ which gives the observed results with highest probability~\cite{PhysRevLettSmolin}, and use $\tilde{T}$ in place of $T$.) Note that since  $\Tr\rho_m=\sum_{i=1}^m r_i$ and $\Tr \sigma_m^\rho=\sum_{i=1}^m\sigma_{ii}$, the TFB are completely determined by $\{r_i\}$ and $\{\sigma_{ij}\}$. Explicitly, the TFB are computed via $\FN(\rho_m,\sigma_m^\rho)=\sum_i \sqrt{\lambda_i}$ and $\FN_\ast(\rho_m,\sigma_m^\rho)=\sum_i \sqrt{\lambda_i}+\sqrt{(1-\sum_{i} r_i)(1-\sum_{i}\sigma_{ii})}$ with $\lambda_i$ the eigenvalues of $T$ and $i=1,\ldots,m$.

Technically, the aforementioned procedure computes  $\FN(\rho_m',\sigma_m^{\rho'})$ and $\FN_\ast(\rho_m',\sigma_m^{\rho'})$. If $C\approx 0$, then $\rho'\approx\rho$ and these quantities are actually bounds on $\FN(\rho,\sigma)$. However, if $C$ is appreciable then to produce certified bounds we must account for  $\rho'\neq\rho$. Here we present two such certified bounds.

\begin{proposition}
The following certified cost function bounds (CCFB) hold:
\begin{equation}
    \FN(\rho'_m,\sigma_m^{\rho'})-\delta\leq \FN(\rho,\sigma)\leq\FN_\ast(\rho'_m,\sigma_m^{\rho'})+\delta\,,
\label{EqnCCFB}
\end{equation}
with $\delta=\min(\delta_1,\delta_2)$, and where $\delta_1=(4 \rank(\rho) C)^{1/4}$ and  $\delta_2=(2\epsilon_m'+\sqrt{2 m C})^{1/2}$ with $\epsilon_m'=1-\Tr\rho'_m$.
\label{prop4}
\end{proposition}

The CCFB show the operational meaning of $C$, which not only bounds the VQSD eigenvalue and eigenvector error \cite{larose2018} but also the VQFE error. Note that $\delta_2$ is directly computable from the VQSD experimental data, whereas $\delta_1$ is useful if one has a promise that $\rho$ is low rank. Alternatively one can use the following certified bounds based on the triangle inequality.

\begin{figure}[t]
    \centering
    \includegraphics[width=\columnwidth]{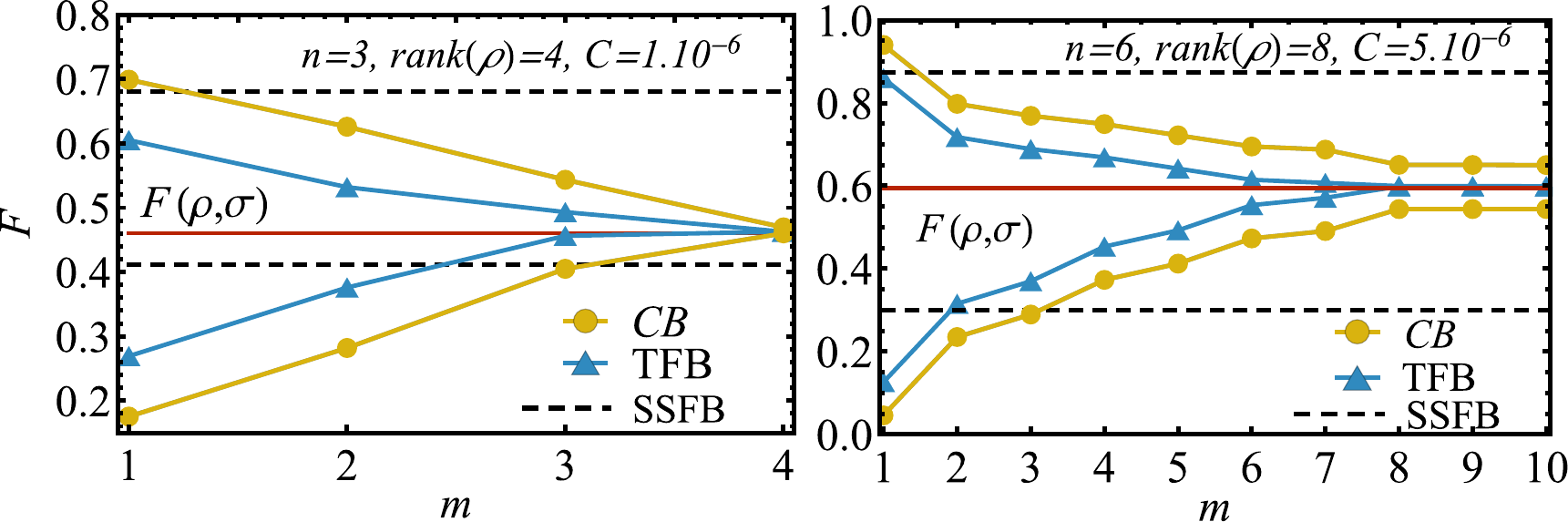}
    \caption{Implementations of VQFE on IBM's quantum computer simulator. The left (right) panel shows bounds on $\FN(\rho,\sigma)$ (solid straight line) versus $m$ for a random state $\sigma$ and a tensor product state $\rho$ of $n=3$ $(n=6)$ qubits with $\rank(\rho)=4$ $(\rank(\rho)=8)$. Dashed lines depict the SSFB. In both cases the TFB converge to $\FN(\rho,\sigma)$. The Certified Bounds (CB) become tighter than the SSFB in both cases, but for  $n=6$ case the certified bounds remain gapped for large $m$ since the cost $C$ is non-negligible.}
    \label{fig:bounds}
\end{figure}

\begin{proposition}
Let $\DN\left(\FN(\rho,\sigma)\right)$ be a distance measure that is monotonically decreasing in $\FN(\rho,\sigma)$. Let $\rho$, $\sigma$, and $\rho'$ be three arbitrary quantum states. Then the following certified triangular inequality bounds (CTIB) hold:
\begin{equation}
    \DN^{\text{LB}}_m\leq \DN\left(\FN(\rho,\sigma)\right)\leq \DN^{\text{UB}}_m\,,
    \label{EqnCTIB}
\end{equation}
with
\begin{eqnarray}
\DN^{\text{LB}}_m=\DN\left(\FN_\ast(\rho_m',\sigma_m^{\rho'} )\right) - \DN\left(\FN(\rho_m',\rho_m^{\rho'} )\right)\,, \label{EqnBoundsDistance21}   \\ 
\DN^{\text{UB}}_m=\DN\left(\FN(\rho'_m,\sigma_m^{\rho'})\right) + \DN\left(\FN(\rho'_m,\rho_m^{\rho'})\right) \,,
    \label{EqnBoundsDistance2}
\end{eqnarray}
where $\rho_m'$, $\sigma_m^{\rho'}$, and $\rho_m^{\rho'}$ are projections onto the subspace of  the $m$ largest eigenvectors of $\rho'$ analogous to (\ref{EqnRhoTruncated}).
\label{prop5}
\end{proposition}

One obtains bounds on $\FN(\rho,\sigma)$ from Proposition \ref{prop5} by inverting $\DN$. Moreover, from Proposition \ref{prop2}, one can show that these bounds are refinable: $\DN^{\text{LB}}_m$ monotonically increases in $m$, and $\DN^{\text{UB}}_m$ monotonically decreases in $m$. The CTIB are valid for the Bures angle $\DN_A\left(\FN(\rho,\sigma)\right)=\arccos\FN(\rho,\sigma)$, Bures  distance $\DN_B\left(\FN(\rho,\sigma)\right)=\sqrt{2-2\FN(\rho,\sigma)}$, and the Sine distance $\DN_S\left(\FN(\rho,\sigma)\right)=\sqrt{1-\FN(\rho,\sigma)^2}$ \cite{Gilchrist2005distance,rastegin2006sine}, and hence one can take the metric that gives the tightest bounds. Furthermore, combining Propositions~\ref{prop4} and \ref{prop5}, we use the term Certified Bounds (CB in Fig.~\ref{fig:bounds}) to refer to the minimum (maximum) of our certified upper (lower) bounds. In other words, the Certified Bounds are obtained by taking the tighter of the two bounds provided by Propositions~\ref{prop4} and \ref{prop5}.
 

\section{Implementations}

Figure~\ref{fig:bounds} shows our VQFE implementations on IBM's quantum computer simulator. The left and right panels show representative results for $n=3$ and $n=6$ qubits, respectively. We used the Constrained Optimization By Linear Approximation (COBYLA) algorithm \cite{PowellCOBYLA} in the VQSD optimization loop   and achieved a cost of $C\!\sim\!10^{-6}$ (see Ref.~\cite{kubler2019adaptive} for a comparison of state-of-the-art optimizers for variational algorithms). We chose $\sigma$ as a random state and $\rho = \bigotimes_{j=1}^n \rho_j$ as a tensor product state, where the latter can be diagonalized with a depth-one quantum circuit ansatz. As one can see, as $m$ increases the TFB rapidly converge to $\FN(\rho,\sigma)$, and since $C$ is small ($\!\sim\!10^{-6}$), the TFB can be viewed as bounds on $\FN(\rho,\sigma)$. Nevertheless we also show our Certified Bounds, and in both cases the Certified Bounds are significantly tighter than the SSFB. 

\begin{figure}[t]
    \centering
    \includegraphics[width=\columnwidth]{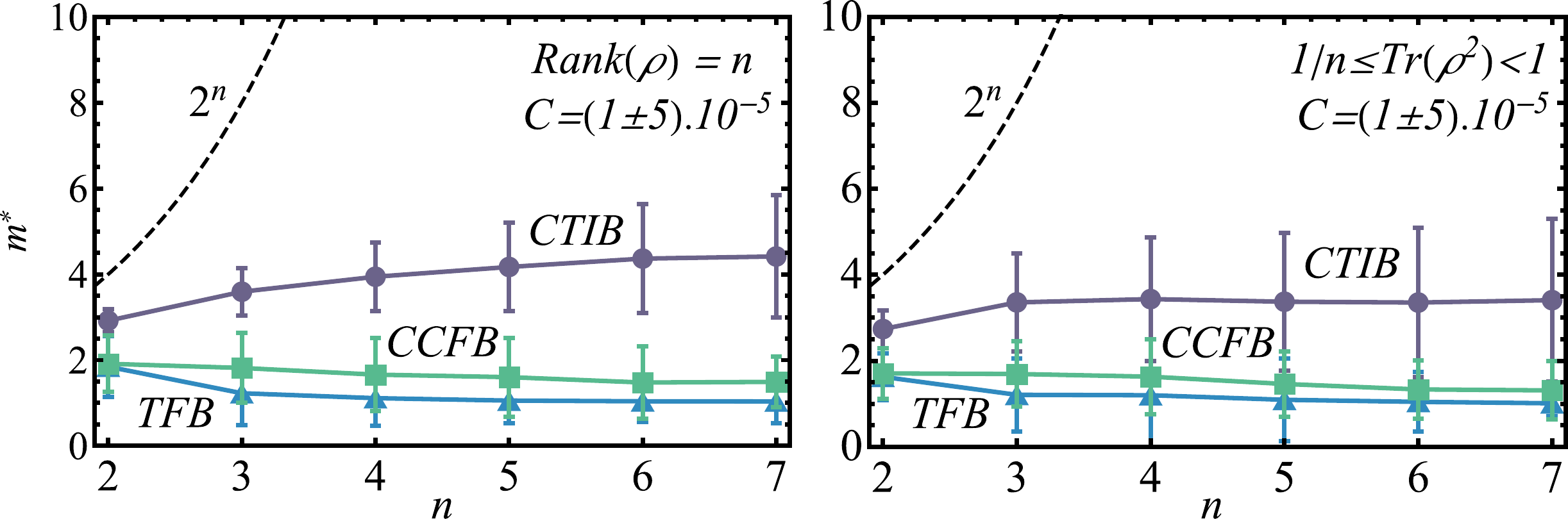}
    \caption{Minimum value of $m$ needed for our bounds (TFB, CCFB, and CTIB) to become tighter than the SSFB ($m^*$) versus $n$. Results were averaged over $2000$ random states $\sigma$ (with uniformly distributed rank and purity), and random states $\rho$ with: low rank ($\rank(\rho)=n$, left), or high purity ($1/n\leq\Tr(\rho^2)<1$, right). Error bars depict standard deviation. A non-zero cost $C$ was obtained by applying a random unitary close to the identity to the diagonal form of $\rho$. As $n$ increases, $m^*$ remains $O(n)$ while the dimension grows exponentially (dashed curve).}
    \label{fig:scaling}
\end{figure}

\section{Heuristic scaling}

Let $m^*$ denote the minimum value of $m$ needed for our bounds to become tighter than the SSFB. Figure~\ref{fig:scaling} plots $m^*$ for the TFB, CCFB, and CTIB for systems with $n=2,\ldots,7$ qubits. The results were obtained by averaging $m^*$ over $2000$ random states $\rho$ and $\sigma$. We considered two cases of interest: when $\rho$ is a low-rank state, and when it has an exponentially decaying spectrum leading to full rank but high purity. In both scenarios $m^*$ is $O(n)$ with $m^* \approx 2$ for the TFB and CCFB, implying that our bounds can outperform the SSFB by only considering a number of eigenvalues in the truncated states which do not scale exponentially with $n$.

\section{Example: Quantum phase transition}

It has been shown that the fidelity can capture the geometric distance between thermal states of condensed matter systems nearing phase transitions and can provide information about the zero-temperature phase diagram  
\cite{Zanardi2007Mixed,Zanardi2007Information,Quan2009Thermal}. We now demonstrate the application of  VQFE  to the study of phase transitions. Consider a cyclic Ising chain of $N=8$ spins-$1/2$ in a uniform field. The Hamiltonian reads
$H=-\sum_{j} h S_j^z+ JS^x_j S^x_{j+1}$, 
with $S_j^{x,z}$ the spin components at site $j$, $h$ the magnetic field, and $J$ the exchange coupling strength. While this model is exactly solvable, it remains of interest as its thermal states can be exactly prepared on a quantum computer \cite{verstraete2009quantum,CerveraLierta2018exactisingmodel}.

Figure \ref{fig:ising} shows the fidelity spectrum for this example. We first verify that the SSFB and the TFB present a pronounced dip near $h=1$, implying that they  can detect the presence of the zero-temperature transition. Moreover, VQFE allows a better determination of the critical field: The TFB give a range for the transition which monotonically tightens as $m$ increases and outperforms the SSFB already for $m=3$ (see inset). The TFB also provide information regarding the closeness of eigenvectors (e.g., upper-TFB which are $\approx 1$ for $m=1,2$), and can detect level crossings where the structure of the subspace spanned by $\{\ket{r_i}\}$ drastically changes. For instance, near $h=1$ the $m=3,4$ TFB present a discontinuity from the crossing between a uniform eigenstate and a pair of exactly degenerate non-uniform symmetry-breaking states. Hence, the fidelity spectrum provides information about the structure of the states beyond the scope of the SSFB or even $\FN(\rho,\sigma)$.

\begin{figure}[t]
    \centering
    \includegraphics[width=\columnwidth]{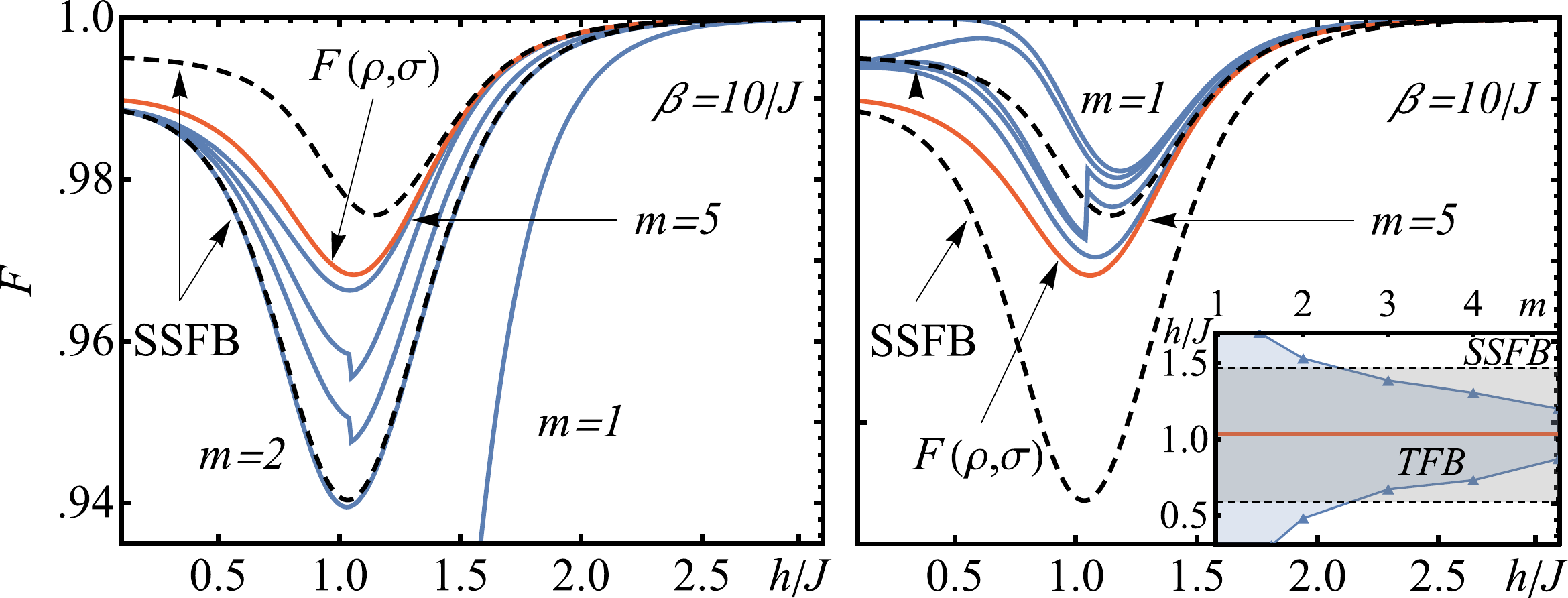}
    \caption{Fidelity spectrum ($m=1,\ldots,5$) between   thermal states $\rho(h)$ and $\rho(h+\delta h)$ of a cyclic Ising chain with $N=8$ spins-$1/2$ in a magnetic field $h$ at inverse temperature $\beta=10/J$, with $J$ the coupling strength and $\delta h=0.01/J$. Dashed lines indicate the SSFB. The left (right) panel depicts the lower (upper) TFB of (\ref{EqnBounds}). The fidelities present a dip at $h\approx 1$ indicating the zero temperature transition of the model. Inset: Bounds for the position of the minimum of the fidelity  (and hence the transition) derived from the SSFB and TFB.}
    \label{fig:ising}
\end{figure}

\section{Complexity analysis}

Recent dequantization results \cite{tang2018quantum, gilyen2018quantum} suggest that exponential speedup over classical algorithms can disappear under low-rank assumptions. Here we show this is not the case for Low-Rank Fidelity Estimation. Low-Rank Fidelity Estimation remains classically hard despite the restriction that one or both of the states are low rank. We formally define Low-Rank Fidelity Estimation as follows.\ \\
\ \\
\textbf{Problem 1.} (\textsc{Low-rank Fidelity Estimation}): \ \\
\textit{Input:} Two $\poly(n)$-sized quantum circuit descriptions $U$ and $V$ that prepare $n$-qubit states $\rho$ and $\sigma$ on a subset of qubits all initially in the $\ket{0}$ state, and a parameter $\gamma$.\ \\
\textit{Promise:} The state $\rho$ is low rank (i.e., $\rank(\rho)= \poly(n)$) or $\epsilon$-low rank (i.e., $\rank_{\epsilon}(\rho) = \poly(n)$ for $\epsilon = 1/\poly(n)$).\ \\
\textit{Output:} Approximation of $F(\rho,\sigma)$ up to additive accuracy $\gamma$.\ \\
\ \\
We remark that VQFE does not require knowledge of the circuit descriptions of $\rho$ and $\sigma$. Rather, allowing for such knowledge is useful when considering classical algorithms.

Next, we analyze the hardness of Low-Rank Fidelity Estimation. Recall that the complexity class \DQC consists of all problems that can be efficiently solved with bounded error in the one-clean-qubit model of computation~\cite{knill1998POOQ}.  The classical simulation of \DQC is impossible unless the polynomial hierarchy collapses to the second level~\cite{DBLP:journals/corr/FujiiKMNTT14, 1704.03640}, which is not believed to be the case. Hence, given the following proposition, one can infer that a classical algorithm {\it cannot efficiently perform Low-Rank Fidelity Estimation}.

\begin{proposition}
\label{prop6}
The problem \textsc{Low-rank Fidelity Estimation} to within precision $\pm \gamma = 1/\poly(n)$ is \DQC-hard.
\end{proposition}

On the other hand, let us now consider the complexity of the VQFE algorithm. The following proposition demonstrates that VQFE is efficient, so long as the VQSD step is efficient.

\begin{proposition}
\label{prop7}
Let $\rho$ and $\sigma$ be quantum states, and suppose we have access to a subroutine that diagonalizes $\rho$ up to error bounded by $C = D_{\HS}(\rho, \rho')$. Let $m$ and $\zeta$ be parameters. Then, VQFE runs in time $\widetilde O(m^6/\zeta^2)$ and outputs an additive $\pm \gamma$-approximation of $F(\rho,\sigma)$, for \begin{equation}
   \gamma = (\epsilon + \sqrt{mC})^{1/2}+ (2 \epsilon + (2+\sqrt{2})\sqrt{mC})^{1/2}+ \zeta \|T\|\,,
\end{equation}
and where $\epsilon$ is determined by $m = \rank_{\epsilon}(\rho)$.
\end{proposition}
\noindent Consider the implications of Proposition \ref{prop7} when $\rho$ is $\epsilon$-low rank, meaning that $\epsilon = 1/\poly(n)$ for $m = \poly(n)$. In this case, VQFE can solve Low-Rank Fidelity Estimation with precision $\gamma = 1/\poly(n)$, assuming that the diagonalization error $C$ is small enough such that $mC = 1/\poly(n)$ and also that $\zeta = 1/\poly(n)$, which suffices since $\|T\| \leq 1$. Under these assumptions, the run time of VQFE will be $\poly(n)$.

\section{Conclusion}

In this work, we introduced novel bounds on the fidelity based on truncating the spectrum of one of the states, and we proposed a hybrid quantum-classical algorithm to compute these bounds. We furthermore showed that our bounds typically outperform the sub- and super-fidelities, and they are also useful for detecting quantum phase transitions. Our algorithm will likely find use in verifying and characterizing quantum states prepared on a quantum computer.

Let us recall that we have proposed employing maximum likelihood methods to reconstruct the $T$ matrix when noise makes it non-positive. However, in the context of state tomography, such procedures can can lead to biases that change the entanglement properties of the estimated state~\cite{schwemmer2015systematic,sentis2018bound}. Hence, we leave for future work the analysis of the effect that such reconstruction methods can have on our bounds.

A strength of VQFE is that it does not require access to purifications of $\rho$ and $\sigma$, although future research could study whether having such access could simplify the estimation of $F(\rho, \sigma)$ (e.g., using Uhlmann's theorem~\cite{nielsen2010}). Another important future direction is to use VQFE to compute distance between quantum operations, e.g., as in~\cite{puchala2011experimentally}. In addition it would be of interest to extend VQFE to the sandwiched Renyi relative entropies $D_{\alpha}(\rho || \sigma)$ \cite{muller2013quantum,wilde2014strong}, defined by: 
\begin{equation}
\label{eqnRenyiFamily}
    2^{(\alpha - 1)D_{\alpha}(\rho || \sigma)} =  \Tr \left(\rho^{\frac{1-\alpha}{2\alpha}} \sigma \rho^{\frac{1-\alpha}{2\alpha}} \right)^{\alpha}\,. 
\end{equation}
Note that $\alpha = 1/2$ corresponds to $F(\rho,\sigma)$. Expanding~\eqref{eqnRenyiFamily} in the eigenbasis of $\rho$, one can see that $D_{\alpha}(\rho || \sigma)$ depends only on $\{r_i\}$ and $\{\sigma_{ij}\}$, i.e., the same quantities used in VQFE. Hence it may be possible to use the same strategy as in VQFE to estimate the entire Renyi relative entropy family.

\section{Acknowledgements}

We thank John Watrous and Mark Wilde for helpful correspondence. MC was supported by the Center for Nonlinear Studies at Los Alamos National Laboratory (LANL). AP was supported by AFOSR YIP award number FA9550-16-1-0495, the Institute for Quantum Information and Matter, an NSF Physics Frontiers Center (NSF Grant PHY-1733907), and the Kortschak Scholars program.  LC was supported by the DOE through the J. Robert Oppenheimer fellowship. PJC acknowledges support from the LANL ASC Beyond Moore's Law project. MC, LC, and PJC also acknowledge support from the LDRD program at LANL. This work was supported in part by the U.S. Department of Energy, Office of Science, Basic Energy Sciences, Materials Sciences and Engineering Division, Condensed Matter Theory Program.

\bibliographystyle{unsrtnat}


\newpage
\clearpage
\setcounter{equation}{0}
\setcounter{figure}{0}
\setcounter{lemma}{0}
\onecolumngrid
\appendix
\begin{center}
\vspace*{\baselineskip}
{\textbf{\large Appendix}}\\
\end{center}

Here we provide detailed proofs of the Propositions presented in the main text. In order to make the results clearer we also reiterate these Propositions. 

\section{Proof of Proposition~\ref{prop1}} \label{proofprop1SM}
\begin{propositionSM}\label{prop1SM}
The following truncated fidelity bounds (TFB) hold:
\begin{equation}
    \FN(\rho_m,\sigma_m^\rho)\leq \FN(\rho,\sigma)\leq\FN_\ast(\rho_m,\sigma_m^\rho)\,.
\label{EqnBoundsSM}
\end{equation}
\end{propositionSM}
\begin{proof}
The right-hand-side of \eqref{EqnBoundsSM} follows from the fact that the generalized fidelity is monotonous under completely positive trace non-increasing (CPTNI) maps  \cite{tomamichel2010duality,tomamichel2015quantum}. In particular, we can define the CPTNI map
\begin{equation}
\EC_m(\rho)=\Pi_m^\rho\rho \Pi_m^\rho \,,
\end{equation}
with $\Pi_{m}^\rho$ the projector onto the subspace spanned by the eigenvectors with $m$-largest eigenvalues, which leads to 
\begin{eqnarray}
 \FN(\rho,\sigma)= \FN_\ast(\rho,\sigma)&\leq& \FN_\ast(\EC_m(\rho),\EC_m(\sigma))\nonumber\\
                    &=& \FN_\ast(\rho_m,\sigma_m^\rho)\,.
\end{eqnarray}
The lower bound is derived from the strong concavity of the fidelity \cite{nielsen2010}, as follows:
\begin{eqnarray}
 \FN(\rho,\sigma)=\FN(\Pi_m^\rho\rho\Pi_m^\rho+\overline{\Pi}_m^\rho\rho \overline{\Pi}_m^\rho,\sigma)&\geq&   \sqrt{p_m}  \FN\left(\frac{\Pi_m^\rho\rho\Pi_m^\rho}{p_m},\sigma\right) +\sqrt{1-p_m} \FN\left(\frac{\overline{\Pi}_m^\rho\rho \overline{\Pi}_m^\rho}{1-p_m},\sigma\right) \nonumber\\
 &\geq&\sqrt{p_m}  \FN\left(\frac{\rho_m}{p_m},\sigma\right) = \FN(\rho_m,\sigma)= \FN(\rho_m,\sigma_m^\rho)\,,\label{EqnProofTruncFid}
\end{eqnarray}
where $p_m=\Tr\rho_m$ and where we have expressed $\rho=\Pi_m^\rho\rho\Pi_m^\rho+\overline{\Pi}_m^\rho\rho \overline{\Pi}_m^\rho$ with $\overline{\Pi}_m^\rho$ the orthogonal complement of $\Pi_m^\rho$. 
\end{proof}

\section{Proof of Proposition~\ref{prop2}} \label{proofprop2SM}

\begin{propositionSM}
The truncated fidelity $\FN(\rho_m,\sigma_m^\rho)$ is monotonically increasing in $m$, and the truncated generalized fidelity $\FN_{*}(\rho_m,\sigma_m^\rho)$ is monotonically decreasing in $m$.
\end{propositionSM}
\begin{proof}
The monotonicity of the truncated fidelity follows from the same procedure as the one used to  derived  (\ref{EqnProofTruncFid}). Namely, we expand   $\rho_{m}=\Pi_{m-1}^\rho\rho\Pi_{m-1}^\rho+\ket{r_m}\rho_{mm} \bra{r_m}$, with $\rho_{mm}=\matl{r_m}{\rho}{r_m}$, and we define $\hat{\rho}_m = \rho_m / p_m$ with $p_m = \Tr \rho_m$. Applying strong concavity gives
\begin{eqnarray}
 \FN(\rho_m,\sigma_m^\rho) = \sqrt{p_m}\FN(\hat{\rho}_m,\sigma ) &\geq& \sqrt{p_m}\left(   \sqrt{p_{m-1} / p_{m}}  \FN\left(\frac{\Pi_{m-1}^\rho\rho\Pi_{m-1}^\rho}{p_{m-1}},\sigma \right) +\sqrt{\rho_{mm}/p_m} \FN\left(\dya{r_m},\sigma \right) \right) \nonumber\\
 &\geq& \FN(\rho_{m-1},\sigma ) = \FN(\rho_{m-1},\sigma_{m-1}^\rho)\,.
\end{eqnarray}
On the other hand, since  $\rho_m=\EC_m(\rho_{m+1})$ and  $\sigma_m^\rho=\EC_m(\sigma_{m+1}^\rho)$, then by the monotonocity under CPTNI maps of the generalized fidelity we get $\FN_\ast(\rho_{m+1},\sigma_{m+1}^\rho)\leq \FN_\ast(\rho_m,\sigma_{m}^\rho)$.
\end{proof}

\section{Proof of Proposition~\ref{prop4}} \label{proofprop4SM}

Before proving Proposition~\ref{prop4}, we first prove four useful lemmas. Of particular interest is Lemma~\ref{lemmaMTSM} below since it generalizes the main result of \cite{coles2019strong} (corresponding to the special case of $\epsilon =0$). All of these lemmas employ our notion of $\epsilon$-rank. We remind the reader that this is defined as follows:
\begin{equation}
\rank_{\epsilon}(\rho) = \min \{  m \in \{1,..., d\} : \|\rho - \rho_m\|_1 \leq \epsilon  \}\,,  
\label{epsilonrankSM}
\end{equation}
where $d$ is the Hilbert-space dimension. 

\begin{lemma}\label{lemma1}
For any positive semi-definite operators $A$, $B$, and $C$, where $C$ is normalized ($\Tr(C) = 1$), we have
\begin{equation}
   |\FN(A,C)-\FN(B,C)|\leq\sqrt{\Tr \left[\left(\sqrt{A} - \sqrt{B}\right)^2 \right]}\,. 
   \label{lemma1SM}
\end{equation}
\end{lemma}
\begin{proof}
We first write the left-hand-side as
\begin{align}
    \Delta\FN=\FN(A , C)-\FN(B , C)= \left\|\sqrt{A}\sqrt{C}\right\|_1 -\left\|\sqrt{B}\sqrt{C}\right\|_1\,.
    \label{DeltaFSM}
\end{align}
It is well known that given a $d\times d$ matrix $M$ with singular value decomposition $M = U_0SU_1$, then $\|M\|_1=|\underset{V\in \UC_d}{\max}\,\,\Tr (M V)|=\Tr(M V_{\text{opt}})$ with $\UC_d$ the set of $d\times d$ unitary matrices and $V_{\text{opt}}=U_1\ad U_0\ad$. Hence, for every pair of $d\times d$ positive semi-definite operators $A$ and $C$ there exists $V_{\text{opt}}$ such that $\FN(A , C)=|\underset{V\in\UC_d}{\max}\,\,\Tr(\sqrt{A}\sqrt{C}V)|=|\Tr(\sqrt{A}\sqrt{C}V_{\text{opt}})|$. (We remark that in the last expression the absolute value is not necessary but we include it for clarity.) Let $V_{\text{opt}}$ and $W_{\text{opt}}$ be the optimal unitaries for $F(A,C)$ and $F(B,C)$, respectively. Then \eqref{DeltaFSM} can be expressed as 
\begin{eqnarray}
    \Delta\FN   &=& \left|\Tr(\sqrt{A}\sqrt{C}V_{\text{opt}})\right|-\left|\Tr(\sqrt{B}\sqrt{C}W_{\text{opt}})\right|\,,\\
    &\leq& \left|\Tr(\sqrt{A}\sqrt{C}V_{\text{opt}})\right|-\left|\Tr(\sqrt{B}\sqrt{C}V_{\text{opt}})\right|\,.\\
               &\leq& \left|\Tr(\sqrt{A}\sqrt{C}V_{\text{opt}})-\Tr(\sqrt{B}\sqrt{C}V_{\text{opt}})\right|\,,\\\
                &=& \left|\Tr\left[(\sqrt{A}-\sqrt{B})\sqrt{C}V_{\text{opt}}\right]\right|\,.
\end{eqnarray}
where we have used  $|\Tr\sqrt{B}\sqrt{C}W_{\text{opt}}|\geq|\Tr\sqrt{B}\sqrt{C}V_{\text{opt}}|$ and  the triangular inequality for the absolute value. 
Then, by means of the matrix H\"{o}lder inequality \cite{baumgartner2011inequality} and the fact that  $C$ is normalized,  we obtain
\begin{align}
    \Delta\FN   \leq \sqrt{\Tr\left[(\sqrt{C}V_{\text{opt}})\ad \sqrt{C}V_{\text{opt}}\right]} \sqrt{\Tr\left[\left(\sqrt{A}-\sqrt{B}\right)^2\right]} = \sqrt{\Tr \left[\left(\sqrt{A} - \sqrt{B}\right)^2 \right]}\,.
\label{eqn500005}
\end{align}
By symmetry, one can replace $\Delta\FN$ with $-\Delta\FN$ in \eqref{eqn500005}. So \eqref{eqn500005} also holds for $|\Delta\FN |$.

\end{proof}

\begin{lemma}\label{lemma4}
Let $\rho$ be a quantum state and let $\rho'=U\ad \ZC(U \rho U\ad) U$ be its VQSD approximation defined in \eqref{EqnRhoPrime} of the main text. Then 
\begin{equation}
\rank_\epsilon(\rho)\leq \rank_\epsilon(\rho')\,.
\end{equation}
\end{lemma}
\begin{proof}
Let us first define $\rhot_Z=\ZC(\rhot)$, with  $\rhot= U\rho U\ad$ and $\ZC$ the dephasing channel in the computational basis, such that $\rho'=U\ad \rhot_Z U$. Due to Schur-Horn's theorem we have that the eigenvalues of $\rhot$ majorize its diagonal elements, implying that $\rhot$ majorizes $\rhot_Z$. Moreover, since $\rhot$ has the same eigenvalues as $\rho$, and $\rho'$ has the same eigenvalues as $\rhot_Z$, then we have that $\rho$ majorizes $\rho'$. In summary,
\begin{equation}
    \rhot \succ \rhot_Z,\quad\text{and}\quad\rho \succ \rho'\,.
\end{equation}
Next, let us define $\hat{m}'=\rank_\epsilon(\rho')$. Since $\rho\succ\rho'$, we have $\sum_{j\leq\hat{m}'}r_j \geq \sum_{j\leq\hat{m}'}r_j'$ and hence 
\begin{equation}
\|\rho - \rho_{\hat{m}'}\|_1 = 1-\sum_{j\leq\hat{m}'}r_j\leq 1-\sum_{j\leq\hat{m}'}r_j'= \|\rho' - \rho_{\hat{m}'}'\|_1 \leq \epsilon.
\end{equation}
The fact that $\|\rho - \rho_{\hat{m}'}\|_1 \leq \epsilon$ implies that $\rank_{\epsilon}(\rho) \leq \hat{m}'$ and hence
\begin{equation}
    \rank_\epsilon(\rho)\leq \rank_\epsilon(\rho')\,.
\end{equation}
\end{proof}

\begin{lemma}\label{lemma2}
Let $\Delta = \rho - \sigma = \Delta^+ - \Delta^-$. Here $\Delta^+$ and $\Delta^-$ respectively correspond to the positive and negative part of $\Delta$, with $\Delta^+\geq0$, $\Delta^-\geq0$, and $\Delta^+\Delta^-=0$. Then \begin{equation}
\rank_{\epsilon}(\Delta^{+})\leq \rank_{\epsilon}(\rho)\,,\quad\text{and}\quad \rank_{\epsilon}(\Delta^-)\leq \rank_{\epsilon}(\sigma)\,.
\end{equation}
\label{lemmaRankEpsilon}
\end{lemma}
\begin{proof}
Let $\{r_j\}$, $\{s_j\}$, and $\{\delta_j\}$ respectively denote the eigenvalues of $\rho$, $\sigma$, and $\Delta$, where the eigenvalues in each set are in decreasing order. By means of Weyl's inequality \cite{horn1990matrix,weyl1912asymptotische} applied to $\rho = \Delta + \sigma$ we get
\begin{equation}
    r_j \geq \delta_j + s_d \,,\quad\forall j\,.
\end{equation}
Because $\sigma\geq 0$, we have $s_d\geq 0$, and hence
\begin{equation}
    r_j \geq \delta_j\,,\quad\forall j\,.
    \label{Weyl1}
\end{equation}
Since $\Delta^+$ is the positive part of $\Delta$, then from the previous equation its eigenvalues (which we denote $\{\delta_j^+\}$ and which are in decreasing order) are such that 
\begin{equation}
    r_j \geq \delta_j^+\geq 0\,,\quad \forall j\,.
        \label{Weyl2}
\end{equation}
Let us now define $\hat{m}=\rank_\epsilon(\rho)$. By means of the definition of $\epsilon$-rank \eqref{epsilonrankSM}, we have that
\begin{equation}
\epsilon\geq  \|\rho - \rho_{\hat{m}}\|_1=\sum_{j>\hat{m}}r_j\,.    
\end{equation}
Moreover, from \eqref{Weyl2} we get $\sum_{j>\hat{m}}r_j\geq \sum_{j>\hat{m}}\delta_j^+=\|\Delta^+ - \Delta^+_{\hat{m}}\|_1$, which gives
\begin{equation}
    \epsilon\geq\|\rho - \rho_{\hat{m}}\|_1 \geq\|\Delta^+ - \Delta^+_{\hat{m}}\|_1\,.
\end{equation}
This implies that $\rank_{\epsilon}(\Delta^{+})\leq \hat{m}$ and hence that $\rank_{\epsilon}(\Delta^{+})\leq \rank_{\epsilon}(\rho)$.

Similarly we denote the set of eigenvalues of $\Delta^-$ as $\{\delta_j^-\}$.  From Weyl's inequality we now have  $s_j \geq \delta_j^-$, $\forall j$.  Let $\tilde{m}=\rank_\epsilon(\sigma)$ such that $\epsilon\geq  \|\sigma - \sigma_{\tilde{m}}\|_1=\sum_{j>\tilde{m}}s_j$. Since $\sum_{j>\tilde{m}}s_j\geq \sum_{j>\tilde{m}}\delta_j^-$, we then get
\begin{equation}
\epsilon\geq\|\sigma - \sigma_{\tilde{m}}\|_1 \geq\|\Delta^- - \Delta^-_{\tilde{m}}\|_1\,,    
\end{equation}
which implies $\rank_{\epsilon}(\Delta^{-}) \leq \tilde{m}$ and hence $\rank_{\epsilon}(\Delta^{-}) \leq \rank_{\epsilon}(\sigma)$.
\end{proof}

\begin{lemma}
Let $D_T(\rho,\sigma)=\frac{1}{2}\Tr|\rho-\sigma|$ denote the trace distance and $D_{\HS}(\rho,\sigma)$ the Hilbert-Schmidt distance, then 
\begin{equation}
(D_T(\rho, \sigma) - \epsilon)^2 \leq R_{\epsilon}(\rho, \sigma) D_{\HS}(\rho, \sigma)\,,
\label{eqnHSDandTD}
\end{equation}
where $R_{\epsilon}(\rho, \sigma) = \rank_{\epsilon}(\rho)\rank_{\epsilon}(\sigma)/(\rank_{\epsilon}(\rho)+\rank_{\epsilon}(\sigma))$.
\label{lemmaMTSM}
\end{lemma}

\begin{proof}
Let $\Delta^+_{m_+}$ be a truncated version of $\Delta^+$ obtained by projecting $\Delta^+$ onto the subspace associated with its ${m_+}$-largest eigenvalues such that ${m_+}=\rank_\epsilon(\Delta^+)$. We define $\Delta^-_{m_-}$ analogously, with ${m_-}=\rank_\epsilon(\Delta^-)$. Then, consider the following density matrices (positive semidefinite matrices with trace one):
\begin{equation}
\tau^+_{m_+} = \Delta^+_{m_+} / \Tr(\Delta^+_{m_+})\,, \qquad \tau^-_{m_-} = \Delta^-_{m_-} / \Tr(\Delta^-_{m_-})\,.
\label{taudef}
\end{equation}
Since the purity of a density matrix is lower bounded by the inverse of its rank, we have
\begin{eqnarray}
\Tr\left[(\tau^+_{m_+})^2\right] &\geq& 1/\rank(\Delta^+_{m_+})=1/m_+=1/\rank_\epsilon(\Delta^+) \geq 1/\rank_\epsilon(\rho)\,,\\
\Tr\left[(\tau^-_{m_-})^2\right] &\geq& 1/\rank(\Delta^-_{m_-})=1/m_-=1/\rank_\epsilon(\Delta^-) \geq 1/\rank_\epsilon(\sigma)\,,
\end{eqnarray}
where the right inequality follows from Lemma~\ref{lemma2}. Inserting \eqref{taudef} and summing the two resulting inequalities gives:
\begin{align}
\Tr\left[(\Delta^+_{m_+})^2\right]+\Tr\left[(\Delta^-_{m_-})^2\right] \geq \frac{\Tr(\Delta^+_{m_+})^2}{\rank_\epsilon(\rho)}+\frac{\Tr(\Delta^-_{m_-})^2}{\rank_\epsilon(\sigma)}\,.
\label{eqlemma21}
\end{align}
Let us first consider the left-hand-side of \eqref{eqlemma21}. Since $\Tr\left[(\Delta^+)^2\right]\geq \Tr\left[(\Delta^+_{m_+})^2\right]$ and $\Tr\left[(\Delta^-)^2\right]\geq \Tr\left[(\Delta^-_{m_-})^2\right]$, and using the fact that $D_{\HS}(\rho,\sigma)= \Tr\left[(\Delta^+)^2\right]+\Tr\left[(\Delta^-)^2\right]$, we obtain
\begin{equation}
    D_{\HS}(\rho,\sigma)\geq\Tr\left[(\Delta^+_{m_+})^2\right]+\Tr\left[(\Delta^-_{m_-})^2\right]\,.
    \label{lemma3DHS}
\end{equation}
On the other hand, since $m_+=\rank_\epsilon(\Delta^+)$ and $m_-=\rank_\epsilon(\Delta^-)$, then from the definition of $\epsilon$-rank we get 
\begin{equation}
    \|\Delta^+ - \Delta^+_{m_+} \|_1 = \Tr(\Delta^+) - \Tr(\Delta^+_{m_+}) \leq \epsilon\,,\qquad  \|\Delta^- - \Delta^-_{m_-} \|_1 = \Tr(\Delta^-) - \Tr(\Delta^-_{m_-}) \leq \epsilon\,.
\end{equation}
Rewriting this and using $D(\rho,\sigma)=\Tr(\Delta^+)=\Tr(\Delta^-)$ gives
\begin{equation}
    \Tr(\Delta^+_{m_+})\geq \Tr(\Delta^+)-\epsilon=D(\rho,\sigma)-\epsilon\,,\qquad     \Tr(\Delta^-_{m_-})\geq \Tr(\Delta^-)-\epsilon=D(\rho,\sigma)-\epsilon\,.
    \label{lemma3DT}
\end{equation}
Hence, combining \eqref{eqlemma21} with \eqref{lemma3DHS} and \eqref{lemma3DT}, we have
\begin{align}
D_{\HS}(\rho,\sigma) \geq \left(D_T(\rho,\sigma)-\epsilon\right)^2\left(\frac{1}{\rank_\epsilon(\rho)}+\frac{1}{\rank_\epsilon(\sigma)}\right)\,,
\end{align}
which is equivalent to \eqref{eqnHSDandTD}.
\end{proof}

\begin{propositionSM}\label{prop4SM}
The following certified cost function bounds (CCFB) hold:
\begin{equation}
    \FN(\rho'_m,\sigma_m^{\rho'})-\delta\leq \FN(\rho,\sigma)\leq\FN_\ast(\rho'_m,\sigma_m^{\rho'})+\delta\,,
\label{EqnCCFBSM}
\end{equation}
with $\delta=\min(\delta_1,\delta_2)$, and where $\delta_1=(4 \rank(\rho) C)^{1/4}$ and  $\delta_2=(2\epsilon_m'+\sqrt{2 m C})^{1/2}$ with $\epsilon_m'=1-\Tr\rho'_m$.
\end{propositionSM}
\begin{proof}
Let us first derive the $\delta_1$ bound. By means of Lemma \ref{lemma1} we have $|\FN(\rho,\sigma)-\FN(\rho',\sigma)|\leq\sqrt{\Tr \left[\left(\sqrt{\rho} - \sqrt{\rho'}\right)^2 \right]}=  \sqrt{2\left(1-\Tr\sqrt{\rho}\sqrt{\rho'}\right)}$. Then from the inequality $1-\Tr\sqrt{\rho}\sqrt{\sigma}\leq D_T(\rho,\sigma)$ \cite{Kholevo1972,audenaert2012comparisons}, we get
\begin{equation}
    |\FN(\rho,\sigma)-\FN(\rho',\sigma)|\leq \sqrt{2 D_T(\rho,\rho')}\,.
    \label{DeltaDT}
\end{equation}
Moreover, taking $\epsilon=0$ in Lemma \ref{lemmaMTSM} results in the bound between the trace distance and Hilbert-Schmidt distance from \cite{coles2019strong}:
\begin{equation}
D_T(\rho,\rho')^2 \leq  R(\rho, \rho')\cdot D_{\HS}(\rho,\rho')\,,
\end{equation}
with $R(\rho, \sigma)$ the reduced rank defined as  $R(\rho, \sigma) = \rank(\rho)\rank(\sigma)/(\rank(\rho)+\rank(\sigma))$. By noting that  $R(\rho, \rho')\leq\rank(\rho)$ and recalling that $C=D_{\HS}(\rho,\rho')$ we find 
\begin{equation}
    |\FN(\rho,\sigma)-\FN(\rho',\sigma)|\leq (4 \rank(\rho) C)^{1/4}\,.
    \label{DeltaFd1}
\end{equation}
Applying the TFB of Proposition \ref{prop1SM}  to $\FN(\rho',\sigma)$ yields
\begin{equation}
    \FN(\rho'_m,\sigma_m^{\rho'})\leq \FN(\rho',\sigma)\leq\FN_\ast(\rho_m',\sigma_m^{\rho'}
    )\,.
\end{equation}
Combining this result with  \eqref{DeltaFd1} gives the $\delta_1$ bound in \eqref{EqnCCFBSM}, i.e.,
\begin{equation}
    \FN(\rho'_m,\sigma_m^{\rho'})-(4 \rank(\rho) C)^{1/4}\leq \FN(\rho,\sigma)\leq\FN_\ast(\rho'_m,\sigma_m^{\rho'})+(4 \rank(\rho) C)^{1/4}\,,
\end{equation}

For the $\delta_2$ bound we first apply Lemma~\ref{lemmaMTSM}  to $\rho$ and $\rho'$ and we specialize  
$\epsilon$ to be $\epsilon_{m}'=1-\Tr(\rho_{m}')$, giving
\begin{equation}
    (D_T(\rho, \rho') - \epsilon_m')^2 \leq R_{\epsilon_m'}(\rho, \rho') D_{\HS}(\rho, \sigma)\,.
\end{equation}
Combining this result with  \eqref{DeltaDT} yields
\begin{equation}
    |\FN(\rho,\sigma)-\FN(\rho',\sigma)|\leq \sqrt{2 \epsilon'_m + 2\sqrt{ R_{\epsilon_m'}(\rho, \rho') D_{\HS}(\rho, \rho')}}\,.
    \label{DeltaDT2}
\end{equation}
Then, from Lemma \ref{lemma4} we know that  $\rank_{\epsilon_m'}(\rho)\leq \rank_{\epsilon_m'}(\rho')$, which implies $R_{\epsilon_m'}(\rho, \rho')\leq\rank_{\epsilon_m'}(\rho')/2=m/2$, where the last equality comes from the fact that  $\epsilon_{m}'=\|\rho'-\rho_{m}'\|_1$.  Moreover, since $D_{\HS}(\rho, \rho')=C$, we have
\begin{equation}
    |\FN(\rho,\sigma)-\FN(\rho',\sigma)|\leq \sqrt{2 \epsilon_{m}' + \sqrt{2 m C}}\,.\label{delta2SM}
\end{equation}
Combining this result with the TFB of Proposition \ref{prop1SM} yields the $\delta_2$ bound in \eqref{EqnCCFBSM}.
\end{proof}

\section{Proof of Proposition~\ref{prop5}} \label{proofprop5SM}

We remark that Proposition~\ref{prop5} applies to any three quantum states, although for our purposes we are interested in its application to the states $\rho$, $\sigma$, and $\rho'$ discussed in the main text. Hence, for consistency, we state this proposition for these states, but we note that these states can be arbitrary.

\begin{propositionSM}
Let $\DN\left(\FN(\rho,\sigma)\right)$ be a distance measure that is monotonically decreasing in $\FN(\rho,\sigma)$. Let $\rho$, $\sigma$, and $\rho'$ be three arbitrary quantum states. Then the following certifiable triangular inequality bounds (CTIB) hold:
\begin{equation}
    \DN^{\text{LB}}_m\leq \DN\left(\FN(\rho,\sigma)\right)\leq \DN^{\text{UB}}_m\,,
    \label{EqnCTIBSM}
\end{equation}
with
\begin{eqnarray}
\DN^{\text{LB}}_m=\DN\left(\FN_\ast(\rho'_m,\sigma_m^{\rho'} )\right) - \DN\left(\FN(\rho'_m,\rho_m^{\rho'} )\right)\,,    \\ \DN^{\text{UB}}_m=\DN\left(\FN(\rho'_m,\sigma_m^{\rho'})\right) + \DN\left(\FN(\rho'_m,\rho_m^{\rho'})\right) \,,
    \label{EqnBoundsDistance2CM}
\end{eqnarray}
where $\rho'_m$, $\sigma_m^{\rho'}$, and $\rho_m^{\rho'}$ are projections of $\rho'$, $\sigma$, and $\rho$, respectively, onto the subspace of the $m$-largest eigenvectors of $\rho'$.
\end{propositionSM}
\begin{proof}
Since $\DN\left(\FN(\rho,\sigma)\right)$ is a distance measure, it satisfies the triangular inequality. Applying this  to the states $\rho$, $\sigma$, and $\rho'$ gives
\begin{align}
     \DN\left(\FN(\rho,\sigma)\right)\geq   \DN\left(\FN(\rho',\sigma)\right) - \DN\left(\FN(\rho',\rho)\right)\,,\nonumber\\
    \DN\left(\FN(\rho,\sigma)\right)\leq  \DN\left(\FN(\rho',\sigma)\right)+\DN\left(\FN(\rho',\rho)\right) \,.
    \label{EqnTriangular}
\end{align}
Combining Proposition \ref{prop1} with \eqref{EqnTriangular} and using the monotonicity of $D$ yields \eqref{EqnCTIBSM}.
\end{proof}

\section{Proof of Proposition~\ref{prop6}} \label{proofprop6SM}

\begin{propositionSM}
The problem \textsc{Low-rank Fidelity Estimation} to within precision $\pm \gamma = 1/\poly(n)$ is \DQC-hard.
\end{propositionSM}
\proof{
We reduce from the problem of approximating the Hilbert-Schmidt inner-product magnitude $\Delta_{\text{HS}}$ between two quantum circuits $\tilde{U}$ and $\tilde{V}$ acting on $n$-qubits \cite{QAQC}, where we define
\begin{align}
\Delta_{\text{HS}}(\tilde{U},\tilde{V}) :=\frac{1}{d^2}|\Tr(\tilde{V}^\dagger \tilde{U})|^2,
\end{align}
where $d=2^n$. 

Consider a specific instance of approximating $\Delta_{\text{HS}}$. We are given as input classical instructions to prepare $\poly(n)$-sized quantum circuits $\tilde{U}$ and $\tilde{V}$ on $n$-qubits each, and the task is to approximate $\Delta_{\text{HS}}(\tilde{U},\tilde{V})$ to precision $1/\poly(n)$. Our reduction will identify this problem as a specific instance of \textsc{Low-rank Fidelity Estimation} (Low-Rank Fidelity Estimation) via the Choi-Jamio\l{}kowski isomorphism over the unitary channels,
\begin{eqnarray}
\mathcal{E}_{\tilde{U}}: \mathcal{D}(\mathcal{H}_d) &\rightarrow \mathcal{D}(\mathcal{H}_d), \quad
X \mapsto \tilde{U} X \tilde{U}^\dagger\\
\mathcal{E}_{\tilde{V}}
: \mathcal{D}(\mathcal{H}_d) &\rightarrow \mathcal{D}(\mathcal{H}_d), \quad
X \mapsto \tilde{V} X \tilde{V}^\dagger
\end{eqnarray}
where $\mathcal{D}(\mathcal{H}_d)$ is the space of $d \times d$ dimensional hermitian matrices.
Consider now the $2n$-qubit maximally entangled state,
\begin{align}
\ket{\Phi^+} = \frac{1}{\sqrt{d}}\sum_{\vec{j}} \ket{\vec{j}} \otimes \ket{\vec{j}}=E\ket{\vec{0}}\,,
\end{align}
where $\vec{j} = (j_1, j_2,...,j_n)$ is a binary vector taking values $j_k$ in $\{0,1\}$, and where  $E$ is an efficent unitary entangling gate (e.g., a depth-two circuit composed of Hadamard and CNOT gates), where $\ket{\vec{0}}=\ket{0}^{\otimes 2 n}$. 

A special case of Low-Rank Fidelity Estimation is when $\rho$ and $\sigma$ correspond to the Choi states of $\tilde{U}$ and $\tilde{V}$. In this case, as the input to Low-Rank Fidelity Estimation, we would be given the gate sequences $U = (\tilde{U} \otimes \id)E$ and $V = (\tilde{V} \otimes \id)E$, which respectively prepare the pure states $\rho$ and $\sigma$ as
\begin{align}
\rho &= U\dya{\vec{0}}U\ad = (\tilde{U} \otimes \id)E  \dya{\vec{0}} E^\dagger(\tilde{U}^\dagger \otimes \id),\\
\sigma &=V\dya{\vec{0}}V\ad = (\tilde{V} \otimes \id)E  \dya{\vec{0}} E^\dagger (\tilde{V}^\dagger \otimes \id)\,.
\end{align}
Then, the fidelity between $\rho$ and $\sigma$ is given by
\begin{align}
\FN(\rho,\sigma)  = |\bra{\vec{0}}E\ad (\tilde{V}\ad \tilde{U} \otimes \id)E \ket{\vec{0}}| = \displaystyle \frac{1}{d}|\Tr(\tilde{V}\ad \tilde{U})|.
\end{align}
We can run \textsc{Low-rank Fidelity Estimation} to estimate the above expression to within $\gamma = 1/\poly(n)$ precision, and thus also approximate $\Delta_{\text{HS}}(U,V)$ to within $1/\poly(n)$. Finally, it is known that approximating $\Delta_{\text{HS}}(U,V)$ to within inverse polynomial precision is \DQC-hard~\cite{QAQC}, and hence the result follows.\qed
}

\section{Proof of Proposition~\ref{prop7}} \label{proofprop7SM}

Before proving Proposition~\ref{prop7SM}, we first prove the following useful lemmas.
\begin{lemma}\label{lemma5}
Let $\epsilon_m=\|\rho-\rho_m\|_1$ and $\epsilon_m'=\|\rho'-\rho_m'\|_1$. Then the following bound holds
\begin{equation}
\epsilon_m'\leq \epsilon_m  + \sqrt{m C}\,,
\end{equation}
where $C = D_{\HS}(\rho, \rho')$ and $\rho'$ is given by Eq.~\eqref{EqnRhoPrime} of the main text.
\end{lemma}
\begin{proof}
Consider the difference
\begin{equation}
    \epsilon_m'-\epsilon_m=\|\rho'-\rho_m'\|_1-\|\rho-\rho_m\|_1=(1-\Tr(\rho'_{m}))-(1-\Tr(\rho_{m}))=\sum_{i=1}^{m} (r_i-r_i')\leq \sum_{i=1}^{m} |r_i-r_i'|=\|\vec{\delta}\|_1\,,
\end{equation}
where we have defined the vector $\vec{\delta}=\vec{r}-\vec{r}'$, with $\vec{r}=\{r_1,\ldots,r_{m}\}$ and $\vec{r}'=\{r_1',\ldots,r_{m}'\}$. From the vector norm equivalence, we have that $\|\vec{\delta}\|_1\leq \sqrt{m} \|\vec{\delta}\|_2 $, which then implies
\begin{equation}
    \epsilon_m'-\epsilon_m\leq \sqrt{m} \|\vec{\delta}\|_2\,.
    \label{eqEprime}
\end{equation}
Moreover, as shown in \cite{larose2018}, the VQSD cost $C$ bounds the eigenvalue error as $\|\vec{\delta}\|_2^2=\sum_{i=1}^{m} (r_i-r_i')^2\leq C$. Hence, combining this with Eq.\ \eqref{eqEprime}, we get $\epsilon_m'\leq \epsilon_m  + \sqrt{mC}$.
\end{proof}

The following lemma is an alternative version of Proposition \ref{prop7SM} that may be of interest in itself, particularly when we do not have the promise that $\rho$ is low rank. Note that this lemma does not refer to the rank properties of $\rho$ but rather refers to the rank properties of $\rho'$. The $\epsilon$-rank of $\rho'$ is experimentally measurable and hence one can use the following lemma to guarantee a particular precision even when there is no prior knowledge about the rank properties of $\rho$.

\begin{lemma}\label{lemma6}
Let $\rho$ and $\sigma$ be quantum states, and suppose we have access to a subroutine that diagonalizes $\rho$ up to error bounded by $C = D_{\HS}(\rho, \rho')$. Let $m$ and $\zeta$ be parameters. Then, VQFE runs in time $\widetilde O(m^6/\zeta^2)$ and outputs an additive $\pm \gamma'$-approximation of $F(\rho,\sigma)$, for 
\begin{equation}
   \gamma' = \sqrt{2 \epsilon_{m}' + \sqrt{2 m C}} + \sqrt{\epsilon_{m}'}+\zeta \|T\|\,,
     \label{eqngammadefSM2}
\end{equation}
where $\epsilon_m' = \|\rho' - \rho_m'\|_1$.
\end{lemma}
\begin{proof}
Let us first define the following quantities which will be useful to bound the error in each step of the VQFE algorithm:
\begin{align}
   \Delta_0&= | F(\rho,\sigma) - \hat F(\rho_m',\sigma)|\,,\\
    \Delta_1&= | F(\rho,\sigma) - F(\rho',\sigma)|\,,  &B_1&= \sqrt{2 \epsilon_{m}' + \sqrt{2 m C}}\,,\\
    \Delta_2&=|\FN(\rho',\sigma) - \FN(\rho_m',\sigma)|\,,  &B_2&= \sqrt{\epsilon'_m}\,,\\
   \Delta_3&=| F(\rho'_m,\sigma) - \hat F(\rho'_m,\sigma)|\,,  &B_3&= m \cdot \|T - \hat T\|\,.
\end{align}
where the notation $\hat{a}$ (e.g., in $\hat{T}$) indicates a fine-sampling estimate of the random variable $a$. Note that the following bounds hold:
\begin{equation}
    \Delta_1 \leq B_1,\qquad  \Delta_2 \leq B_2,\qquad  \Delta_3 \leq B_3\,,
\end{equation}
where the first inequality comes from Eq.\ \eqref{delta2SM} and the second inequality can be derived from Eq.\ \eqref{previousprop3} of the main text as follows
\begin{equation}
    \FN(\rho',\sigma) - \FN(\rho_m',\sigma)\leq \FN_\ast(\rho_m',\sigma_m^{\rho'}) - \FN(\rho_m',\sigma_m^{\rho'})\leq \|\rho' - \rho_m'\|_1^{1/2}=\sqrt{\epsilon'_m}\,.
\end{equation}
Finally, we prove the third inequality in Eq.\ \eqref{singular} below.

By means of multiple applications of the triangle inequality, we obtain the following result:
\begin{align}
\Delta_0&\leq  \Delta_1+| F(\rho',\sigma) - \hat F(\rho_m',\sigma)|\\
        &\leq \Delta_1+ \Delta_2+\Delta_3 \label{DBounds} \\
        &\leq B_1+ B_2 + B_3\,. \label{BBounds}
\end{align}
Note that $\Delta_0$ is ultimately the quantity we are interested in bounding since $F(\rho, \sigma)$ is the desired fidelity while $\hat F(\rho_m',\sigma)$ is the quantity that VQFE actually outputs. While \eqref{BBounds} provides an upper bound on $\Delta_0$, we will need to re-formulate this as a probabilistic bound since $B_3$ depends on the estimate $\hat T$ of a random variable.

To do this re-formulation, let us consider the sources of statistical noise on $\hat T$. In the first step of VQFE, we call the VQSD subroutine to approximately diagonalize  $\rho$ up to some error bounded by $C=D_{\HS}(\rho, \rho')$. Here $\rho'=\sum_{i} r_i' \ketbra{r_i'}{r_i'}$, where $\{r_i'\}$ and $\{\ket{r_i'}\}$ are approximations of the eigenvalues and eigenvectors of $\rho$, respectively. While $\rho'$ can only be estimated up to some finite sampling precision, such error can be tuned appropriately. We refer to \cite{larose2018} for further details.

The second step of VQFE employs the Swap Test to measure the matrix elements $\sigma_{ij}=\matl{r_i}{\sigma}{r_j}$. If $C\neq 0$ we measure instead  $\sigma_{ij}'=\matl{r_i'}{\sigma}{r_j'}$, which can be used to  construct the approximate matrix $T'$ with matrix elements  ${T}_{ij}'=(r_i'r_j')^{1/2}\sigma_{ij}'$. For the sake of brevity, we will simply write $T$ instead of $T'$.
For each entry of the symmetric matrix ${T}$ that needs to be determined, we call the Swap Test $O(m^4/\zeta^2 \cdot \log(2m^2/\delta))$ times. We denote $\hat{T}$ as the finite-sampling estimate of $T$.
Then, according to the two-sided Chernoff bound, we have that any entry $(i,j)$ in $\hat T$ has relative error at most $\zeta/m^2$ with probability greater than $1 - \delta/m^2$,
\begin{align}
\Pr\Big[ |{\hat T}_{ij} - T_{ij}| \leq \frac{\zeta }{m^2} \cdot T_{ij} \Big] \geq 1-\delta/m^2.
\end{align}
By the union bound, the probability that this is the case simultaneously for all entries is then given by
\begin{align}\label{union-bound}
\Pr\Big[ \, \forall (i,j) \in [m] \times [m]: \,\,\, |{\hat T}_{ij} - T_{ij}| \leq \frac{\zeta }{m^2} \cdot T_{ij} \Big] \geq 1-\delta.
\end{align}
Then, we conclude that
\begin{align}\label{matrix-error}
&\Pr\Big[ \|  T - {\hat T}  \| \leq \frac{\zeta \|T\|}{m} \Big] \notag\\ 
&=\Pr\left[ \sqrt{\sum_{i,j=1}^m |T_{ij} - \hat T_{ij}|^2} \leq \frac{\zeta \|T\|}{m} \right] \\
&=
\Pr\left[ \sqrt{\sum_{i,j=1}^m |T_{ij} - \hat T_{ij}|^2} \leq \frac{\zeta \|T\|}{m} \hspace{2mm} \Big{|} \hspace{1.8mm} \forall (i,j): \, |{\hat T}_{ij} - T_{ij}| \leq \frac{\zeta  \cdot T_{ij}}{m^2} \right] \cdot \Pr\left[\forall (i,j): \, |{\hat T}_{ij} - T_{ij}| \leq \frac{\zeta \cdot T_{ij}}{m^2} \right]\\
& \geq
\underbrace{\Pr\left[  \frac{\zeta }{m^2}\sqrt{\sum_{i,j=1}^m T_{ij}^2} \leq \frac{\zeta \|T\|}{m} \hspace{2mm} \Big{|} \hspace{1.8mm} \forall (i,j): \, |{\hat T}_{ij} - T_{ij}| \leq \frac{\zeta  \cdot T_{ij}}{m^2} \right]}_{= \, 1} \cdot \Pr\left[\forall (i,j): \, |{\hat T}_{ij} - T_{ij}| \leq \frac{\zeta \cdot T_{ij}}{m^2} \right]\\
&= \Pr\left[ \,\forall (i,j): \, |{\hat T}_{ij} - T_{ij}| \leq \frac{\zeta }{m^2} \cdot T_{ij} \right] \quad \geq  \quad 1 - \delta. \quad \quad (\text{by Eq. } \eqref{union-bound}) \label{2-norm-bound}
\end{align}

The third step of VQFE requires the diagonalization of the $m\times m$ symmetric matrix $\hat T$, which can be accomplished in time $O(m^3)$. Let us now prove the bound $\Delta_3\leq B_3$, i.e., the bound on the finite sampling error in our estimate $\hat F(\rho'_m,\sigma)$ of the fidelity $F(\rho'_m,\sigma)$. By denoting $\lambda_j$ and $\hat{\lambda}_j$ the eigenvalues of $T$ and $\hat{T}$, respectively,  we get
\begin{align} \label{singular}
\Delta_3 & =| F(\rho'_m,\sigma) - \hat F(\rho'_m,\sigma)| \notag \\
& = | \Tr\sqrt{T} - \Tr\sqrt{\hat T}| \notag \\
& \leq \sum_{j=1}^m | \sqrt{\lambda_j} - \sqrt{\hat{\lambda_j}}| \notag\\
& \leq m \cdot \max_{1 \leq j \leq m } | \sqrt{\lambda_j} - \sqrt{\hat{\lambda_j}}| \notag \\
& \leq m \cdot \|T - \hat T\|= B_3\,.
\end{align}
Then, by means of Eq. (\ref{2-norm-bound}) and Eq. (\ref{singular}), we get
\begin{align}
\label{la-bound}
\Pr \left[ \Delta_3  \leq \zeta \|T\| \right] =\Pr \left[ | F(\rho'_m,\sigma) - \hat F(\rho'_m,\sigma)|  \leq \zeta \|T\| \right] &= \Pr \left[ | \Tr\sqrt{T} - \Tr\sqrt{\hat T}| \leq \zeta \|T\| \right]\\
& \geq \Pr \left[ m \cdot \|T - \hat T\|  \leq \zeta \|T\| \right] \\
& \geq \quad 1 - \delta\,. \label{EqDeltaSM}
\end{align}

Let us define 
\begin{equation}
    \gamma' = B_1+B_2+\zeta \|T\| = \sqrt{2 \epsilon_{m}' + \sqrt{2 m C}} + \sqrt{\epsilon_{m}'}+\zeta \|T\| \,,\label{eqngammap}
\end{equation}
and combine our previous results to bound the VQFE run time:
\begin{align}
\Pr \left[ | F(\rho,\sigma) - \hat F(\rho_m',\sigma)| \leq \gamma' \right]
&= \Pr \left[ \Delta_0 \leq B_1+B_2+\zeta \|T\| \right]\\
&\geq \Pr \left[ \Delta_1+\Delta_2+\Delta_3 \leq B_1+B_2+\zeta \|T\| \right]  \quad \quad (\text{by Eq. } \eqref{DBounds}) \\
&\geq \Pr \left[ B_1+B_2+\Delta_3 \leq B_1+B_2+\zeta \|T\|\right]  \quad \quad (\text{by Eq. } \eqref{BBounds}) \\
&= \Pr \left[ \Delta_3\leq \zeta  \|T\|\right]\\
& \geq \quad 1 - \delta. \quad \quad \quad \quad \,\,\, \quad \quad \quad \quad \quad \quad \qquad \quad (\text{by Eq. } \eqref{EqDeltaSM}) \label{Prgammap}
\end{align}

As previously mentioned, VQFE performs the Swap Test for each of the $O(m^2)$ matrix elements.
Therefore, VQFE runs in time $O(m^6/\zeta^2 \cdot \log(2m^2/\delta))$, or just $\widetilde O(m^6/\zeta^2)$ when neglecting slower growing logarithmic contributions. Furthermore, from \eqref{Prgammap}, we have that VQFE outputs an additive $\pm \gamma'$-approximation of the fidelity $F(\rho,\sigma)$ with precision $\gamma'$ defined by \eqref{eqngammap} and with probability at least $1- \delta$.
\end{proof}

\begin{propositionSM}
Let $\rho$ and $\sigma$ be quantum states, and suppose we have access to a subroutine that diagonalizes $\rho$ up to error bounded by $C = D_{\HS}(\rho, \rho')$. Let $m$ and $\zeta$ be parameters. Then, VQFE runs in time $\widetilde O(m^6/\zeta^2)$ and outputs an additive $\pm \gamma$-approximation of $F(\rho,\sigma)$, for \begin{equation}
   \gamma = (\epsilon + \sqrt{mC})^{1/2}+ (2 \epsilon + (2+\sqrt{2})\sqrt{mC})^{1/2}+ \zeta \|T\|\,,
     \label{eqngammadefSM}
\end{equation}
where $\epsilon$ is determined by $m = \rank_{\epsilon}(\rho)$.
\label{prop7SM}
\end{propositionSM}
\begin{proof}
This proposition is very similar to Lemma \ref{lemma6} except that $\gamma'$ is replaced by $\gamma$. To make this replacement we show that $\gamma\geq\gamma'$ as follows: 
\begin{align}
    \gamma&=\sqrt{2 \epsilon + (2 + \sqrt{2})\sqrt{ m C}} + \sqrt{\epsilon+\sqrt{mC}}+\zeta  \|T\|\nonumber\\
   & \geq \sqrt{2 \|\rho - \rho_m\|_1 + (2 + \sqrt{2})\sqrt{ m C}} + \sqrt{\|\rho - \rho_m\|_1+\sqrt{mC}}+\zeta  \|T\|\nonumber\\
    &\geq \sqrt{2 \epsilon_{m}' + \sqrt{2 m C}} + \sqrt{\epsilon_{m}'}+\zeta  \|T\| \nonumber\\
&=\gamma'\,.
\end{align}
The first inequality follows from the fact that $m=\rank_\epsilon(\rho)$ and the second inequality follows from Lemma \ref{lemma5}.

Recall from the proof of Lemma \ref{lemma6} (Eq.\ \eqref{Prgammap}) that the output of VQFE satisfies the following inequality 
\begin{align}
\Pr \left[ | F(\rho,\sigma) - \hat F(\rho_m',\sigma)| \leq \gamma' \right]
  \geq 1 - \delta.
\end{align}
Then, because $\gamma\geq\gamma'$ we also have that
\begin{align}
\Pr \left[ | F(\rho,\sigma) - \hat F(\rho_m',\sigma)| \leq \gamma \right]
\geq \Pr \left[ | F(\rho,\sigma) - \hat F(\rho_m',\sigma)| \leq \gamma' \right] \geq 1-\delta\,.
\label{eqnprop7finalresult}
\end{align}
As noted in the proof of Lemma \ref{lemma6}, VQFE achieves this output with a run time $\widetilde O(m^6/\zeta^2)$. Hence, with this run time, VQFE outputs a $\pm \gamma$-approximation of the fidelity $F(\rho,\sigma)$ with precision $\gamma$ defined by \eqref{eqngammadefSM} and with probability at least $1- \delta$.

\end{proof}

\end{document}